\documentclass[aps,prl,reprint,groupedaddress]{revtex4-1}
\usepackage{amsmath}
\usepackage{amssymb}
\usepackage{epsfig}
\usepackage{textcomp}
\usepackage{amsthm}

\newtheorem{theorem}{Theorem}

\newtheorem{lemma}[theorem]{Lemma}
\newtheorem{corollary}[theorem]{Corollary}

\font\german=eufm10 at 10pt \def\Buchstabe#1{{\hbox{\german #1}}}
\def\EA{\Buchstabe{A}} 

\newcommand{\cC}{\mathcal{C}}

\newcommand{\cM}{\mathcal{M}}

\newcommand{\cS}{\mathcal{S}}

\font\openface=msbm10 at10pt 
  \def\Reals {{\hbox{\openface R}}}

\begin{document}

\title {Bounding quantum contextuality with lack of third-order interference}

\author{Joe Henson\footnote{H.H. Wills Physics Laboratory, University of Bristol, Tyndall Avenue, Bristol, BS8 1TL, U.K.} }
  
\begin{abstract}
Recently many simple principles have been proposed that can explain
quantum limitations on possible sets of experimental probabilities in
nonlocality and contextuality experiments. However, few implications
between these principles are known.
Here it is shown that lack of irreducible third-order interference (a
generalisation of the idea that no probabilistic interference remains
unaccounted for once we have taken into account interference between
\textit{pairs} of slits in a $n$-sit experiment) implies the principle
known as the E principle or Consistent Exclusivity (that, if each pair
of a set of experimental outcomes are exclusive alternatives in some
measurement, then their probabilities are consistent with the
existence of a further measurement in which they are all exclusive).
This is a step towards a more unified understanding of quantum nonlocality and contextuality, which promises to allow derivations of important results from minimal, easily grasped assumptions. As one example, this result implies that lack of third order interference bounds violation of the CHSH-Bell inequality to $2.883$.
\end{abstract}
\maketitle
%=====================================================
It is of great interest to formulate simple principles obeyed by quantum mechanics (QM) from which otherwise mysterious or difficult results can be derived. Such principles can clarify the options when we consider what properties of quantum mechanics are most likely to persist in more developed physical theories. This question has relevance for quantum gravity, where many have considered going beyond standard quantum mechanics in the light of such issues as black hole evaporation and the problems of time (see \textit{e.g.} \cite{Sorkin:1997gi,Hartle:1998yg,Smolin:2006bw}). This has led to a convergence of interests from the study of quantum information and quantum gravity, in which both sides stand to gain new understanding.

In Bell-type nonlocality experiments, QM allows only a specific set of experimental probabilities \cite{Tsirelson:1980, Tsirelson:1987, Tsirelson:1993}, and it is interesting to look for an explanation of this in terms of simple principles \cite{Popescu:1994}.  The same question can be asked for broader classes of ``contextuality scenarios''. As such principles proliferate in the literature \cite{vanDam:2005,Craig:2006ny,Linden:2007,Pawlowski:2009,navascues:2009,Fritz:2012, Cabello:2012}, it becomes increasingly important to search for logical relations between them \cite{Navascues:2014}. The principle of \textit{Consistent Exclusivity} (CE) and the closely-related E-principle (called local orthogonality when applied to nonlocality, and also strongly related to orthomodularity and orthocoherence in earlier literature \cite{Foulis:1981}) is of particular interest \cite{Fritz:2012, Cabello:2012,Fritz:2013}. It is a trivial observation that, given a set of alternative outcomes for a given experiment, the sum of their probabilities is less than one.  CE requires that the probabilities of a set of experimental outcomes should also sum to less than one if each \textit{pair} of outcomes in the set is exclusive in some experiment, which is a stronger requirement when some outcomes in different measurements are considered to be physically identified (see below).  Similarly, it is a defining feature of QM that, in a multiple slit experiment, the probability of the particle reaching a particular region on the screen when two slits are opened may not be the same as the sum of the probabilities when each one of those slits is open.  However, in situations where such interference between \textit{pairs} of alternatives is ruled out, there is no further interference in QM.  Coming from the quantum gravity perspective, Sorkin has argued that a general form of this ``lack of (irreducible) third-order interference'' should be considered the most fundamental property of QM \cite{Sorkin:1994}. This principle has recently been directly tested in three-slit experiments \cite{Sinha:2009}, and has been applied in the context of generalised probabilistic theories as one of a number of postulates from which quantum mechanics can be reconstructed \cite{Ududec:2011,Barnum:2014ysa}, and, in a different framework, to imply some of the same restrictions as has CE \cite{Niestegge:2011, Niestegge:2013} \footnote{\label{f:other_formalisms}The results reported here differ substantially from Niestegge's, and the use of third order interference is quite different. Here, as far as the ban on PR boxes and Wright pentagon states is concerned, the only role of lack of third order interference is to imply CE.  Then the result follows from CE and a composibility assumption as in \cite{Fritz:2012, Cabello:2012}.  Niestegge implicitly assumes CE as part of his basic framework (see \textit{e.g.} page 3 of \cite{Niestegge:2012}, where the comments on the properties of states on ``orthospaces'' easily imply CE), and the same limitations on behaviours are derived without the composibility assumption, by adding a version of lack of third order interference.  The fact that these important properties of QM can be derived in two quite different starting points, one involving independent, separated systems and the other not, is surely significant for reconstructions of QM from simple principles.}. Here it is shown that, in a relatively simple framework for contextuality in the spirit of Sorkin's original idea, lack of third order interference implies CE. Results on a strengthening of this condition are to be found in \cite{Craig:2006ny,Dowker:2013}; here the weakest form of the principle is investigated in a more general setting.

Consider a hypothetical experiment in which different measurements can be chosen, which may be incompatible in the sense that carrying out one may affect the statistics of others. We will allow the identification of particular outcomes of different measurements (a concrete example being the identification of outcomes in QM experiments when they correspond to the same Hilbert space subspace).
Now consider a ``sample space'' $\Xi$, and let us identify every measurement with a partition $M$ of $\Xi$, and every ``fine-grained'' outcome of that measurement with a set $A \in M$. In this way an element of $\Xi$ specifies an outcome for every experiment. The set of all measurements will be called $\cM$. The term \textit{coarse-grained outcomes} for a measurement $M$ will refer to subsets of $M$ (including the empty set). For each $M\in \cM$ the set of all coarse-grained outcomes form a Boolean algebra $\EA_M$. The set of all coarse-grained outcomes across all measurements will be denoted $\cC := 
\bigcup_{M\in \cM} \,\EA_M$. Note that, as desired, the same outcome may appear in two different measurements. The space $\Xi$ together with the set $\cM$ of all measurements specifies a \textit{partition scenario} $\cS=\{\Xi, \cM\}$.  We say that outcomes $A$ and $B$ are \textit{exclusive} if they are disjoint and there exists a measurement $M \in \cM$ such that $A,B \in M$.

To understand this, consider the following ``marginal scenarios'' \cite{Abransky:2011}.  They involve a set of ``boxes'' with labels in $X=\{1,..,n\}$.  When a box is opened it can be found to be empty or full, denoted by the outcome bit $a \in\{0,1\}$.  Only certain subsets of the boxes $J \subset 2^X$ can be jointly opened.  A measurement picks out a subset $j \in J$ of the boxes to open and an outcome of that measurement corresponds to the assignment of a bit to each of those boxes, $s \in 2^j$.  To represent this as a partition scenario, we take the sample space to be all $n$-bit strings, $\Xi=2^X$, so that each string specifies an outcome for every box.   The outcome $A_s \subset \Xi$ comprises all of these strings that agree with the outcomes for the boxes actually measured, $j$: formally, $A_s = \{\gamma \in 2^X : \gamma|_j = s \}$  where $\gamma|_j$ is the restriction of the function $\gamma$ over $X$ to $j$.  The measurement corresponding to subset $j \subset X$ is represented by the partition $M_j=\{A_s\}_{s \in 2^j}$.  Other measurements, in which later choices of box are functions of earlier outcomes, can also be included \footnote{Indeed these ``correlated'' or ``branching'' measurements (see \cite{Dowker:2013} and appendix D of \cite{Fritz:2013}) must be added if the definition of exclusivity given above is to make intuitive sense; otherwise it is possible that two outcomes which contradict each other for a particular box may not appear as exclusive alternatives in one measurement.  Alternatively, exclusivity can be defined so that outcomes $A$ and $B$ are exclusive if there exists a course-grained outcomes $C$ and $D$ such that $A \subseteq C$, $B \subseteq D$, $C$ and $D$ are disjoint and there exists a measurement $M \in \cM$ such that $C,D \in M$.}. If the measurable subsets $j$ are such that exactly one of each subset in a partition of the boxes can be opened, then we have a ``Bell scenario'' (two pairs of boxes, such that only one box in each pair can be opened, is the CHSH scenario).  Another well-known example consists of $n$ boxes such that only pairs labelled $\{i,i+1\}$ for all $i$ and $\{n,1\}$ can be jointly measured.  ``Specker's parable'' concerns the $n=3$ case \cite{Specker:1960}, while for $n=5$ there is a set of outcomes $\{A_i\}$ with $i=1...5$ such that $\{A_i,A_{i+1}\}$ for $i=1...4$ and $\{A_5,A_1\}$ are the only exclusive pairs, known as a ``Wright pentagon'' \cite{Wright:1978}.

Given a partition scenario, a \textit{probability function} $P(\cdot)$ represents a set of experimental results.  Its domain is the set of all outcomes $\cC$, but the function $P$ is only required to be a probability measure when restricted to the outcomes $\EA_M$ for a given measurement $M$; thus the only restriction on experimental probabilities is that identified outcomes have the same probability (``consistency'') \footnote{Other frameworks inspired by the Kochen-Specker theorem can be defined, for example the AFLS formalism \cite{Fritz:2013}.  These cases can be treated in the marginal scenario framework, and thus represented as partition scenarios, via a restriction on probability functions (see section 7 of \cite{Abransky:2011} and appendix D of \cite{Fritz:2013}).  Using this map between the frameworks, it is not difficult to see that CE for the underlying marginal scenario implies CE for the derived ALFS scenario, and thus the main result of this letter extends to the AFLS framework.}.

Turning to restrictions on the experimental probabilities, non-contextuality requires that there exists a \textit{joint probability distribution} $P_J$ on $\Xi$ such that $P_J(A)=P(A)$ $\forall A \in \cC$.  That is, the experimental probabilities can be derived from a probability distribution over the whole sample space. It is well-known that this principle is incompatible with QM.  Consistent Exclusivity (CE) \cite{Cabello:2012, Fritz:2012} can be seen as a weakening of this condition.  A probability function $P$ on a scenario $\cS$ obeys CE if, for all sets $S$ of fine-grained outcomes such that $A$ and $B$ are exclusive for all pairs $\{A,B\} \subset S$,
\begin{equation}
\label{e:ce}
\sum_{A \in S} P(A) \leq 1.
\end{equation}
This definition follows 7.1.1 of \cite{Fritz:2013}, by which CE accords with the E principle, ``the sum of the probabilities of any set of pairwise mutually exclusive events cannot be higher than 1'' \cite{Cabello:2014a}.

Non-contextuality can instead be weakened by replacing the joint probability measure with a generalised measure that, while agreeing with the experimental probabilities, allows interference, meaning violation of the Kolmogorov sum rule. This interference is not unrestricted in QM, however -- otherwise any probability function would be allowed.  A \textit{joint quantum measure} is a function $\mu:2^\Xi \rightarrow \Reals_{\ge 0}$ such that 
\begin{equation}
\label{e:q_dec}
\mu(A)=P(A) \quad \forall A \in \cC,
\end{equation}
and such that for any three disjoint sets $A \subset \Xi$, $B \subset \Xi$ and $C \subset \Xi$, 
\begin{gather} 
\mu(A) + \mu(B) + \mu(C) \nonumber\\ - \mu(A \cup B) - \mu(B \cup C) - \mu(C \cup A) \nonumber\\ + \mu(A \cup B \cup C) = 0. \label{e:q_sum_rule}
\end{gather}
Equation (\ref{e:q_dec}) ensures that the quantum measure $\mu$ reduces to the experimental probabilities $P$ when restricted to measurement outcomes.  This implies that $\mu$ obeys the Kolmogorov rule when restricted to the outcomes of one experiment: we have $P(A)+P(B)=P(A\cup B)$ for exclusive outcomes, which, substituting all three terms by eqn.~(\ref{e:q_dec}), gives $\mu(A)+\mu(B)=\mu(A\cup B)$.  Equation (\ref{e:q_sum_rule}) is known as the \textit{Sorkin sum rule} (or ``quantum sum rule'').  It is not hard to check that both CE and the existence of a joint quantum measure are implied if there is a standard quantum model for the probability function \cite{Fritz:2013,Dowker:2013}.

Given this definition, it might not be obvious why lack of third order interference is being used as a synonym for the existence of a joint quantum measure.  The following lemma clarifies this.
\begin{lemma}
Consider a probability function $P$ on a scenario $\cS$ that admits a joint quantum measure, and consider a partition $Q$ of a set $X \subset \Xi$. If $\mu(A)+\mu(B)=\mu(A \cup B)$ for all $A,B \in Q$ then 
\begin{equation}
\label{e:jointni}
\mu(X)=\sum_{A \in Q} \mu(A).
\end{equation}
\label{l:nohigher}
\end{lemma}
\begin{proof}
For $|Q|=2$ the statement is trivially true. Assume that the lemma holds in all cases with $|Q| \leq n$ for some $n \geq 2$, and consider a partition $Q$ of a set $X$ for which $|Q|=n+1$, obeying the condition $\mu(A)+\mu(B)=\mu(A \cup B)$ for all $A,B \in Q$. Let $A,B \in Q$ be two events in the partition and let $Y=X \backslash (A \cup B)$. Applying the Sorkin sum rule (\ref{e:q_sum_rule}) to $\{Y,A,B\}$ yields
\begin{gather}
\nonumber \mu(Y) + \mu(A) + \mu(B) - \mu(Y \cup A) - \\ \mu(A \cup B) - \mu(Y \cup B) + \mu(X) = 0.
\end{gather}
The inductive hypothesis implies that $\mu(Y \cup B)=\mu(Y)+\mu(B)$, and by assumption $\mu(A \cup B) - \mu(A) = \mu(B)$, and so this implies
\begin{equation}
\mu(X) = \mu(Y \cup A) + \mu(B) = \sum_{C \in Q} \mu(C),
\end{equation}
where the inductive hypothesis has been applied to the set $Y \cup A$ in the last step.
\end{proof}
To sum this lemma up in a slogan, the existence of a joint quantum measure implies that  ``pairwise non-interference implies joint non-interference,'' where joint non-interference refers to eqn.~(\ref{e:jointni}).  This is highly suggestive of a connection to CE.  As explained under eqn.~(\ref{e:q_sum_rule}), if a pair of events is exclusive then it is non-interfering, and hence, given the existence of a joint quantum measure, pairwise exclusivity implies joint non-interference.  However, this does not suffice to show that the existence of a joint quantum measure implies CE; the quantum measure is not in general bounded above, and so some work remains to be done to derive eqn.(\ref{e:ce}).
\begin{theorem}
Consider a probability function $P$ on a scenario $\cS$. If $P$ admits a joint quantum measure then it obeys Consistent Exclusivity.
\label{t:main}
\end{theorem}
\begin{proof}
Assume that the probability function $P$ admits a joint quantum measure, and consider a set of fine-grained outcomes $S \subset \cC$ such that $A$ and $B$ are exclusive for all pairs $\{A,B\} \subset S$. Let us define the sets X= $\bigcup_{A \in S} A$ and $R = \Xi \backslash X$ , and also the set $Q= S \cup \{ R \}$, which is a partition of $\Xi$. Furthermore, using (\ref{e:q_dec}), we have that $\mu(A)+\mu(B)=\mu(A \cup B)$ for all $A,B \in S$. Now, for some $B \in S$, let us apply the Sorkin sum rule (\ref{e:q_sum_rule}) to the sets $Y= X \backslash B$, $B$ and $R$. We obtain
\begin{gather}
\nonumber \mu(Y) + \mu(B) + \mu(R) - \\ \mu(Y \cup B) - \mu(Y \cup R) - \mu(B \cup R) + \mu(\Xi) = 0.
\end{gather}
Applying lemma \ref{l:nohigher} gives $\mu(Y) + \mu(B) = \mu(Y \cup B)$, and we have $\mu(Y \cup R) + \mu(B)= \mu(\Xi)$ from (\ref{e:q_dec}) because $B$ (and thus its complement $Y \cup R$) is a measurement outcome, giving the result
\begin{equation}
\mu(B) + \mu(R) = \mu(B \cup R).
\end{equation}
Because this applies for all $B \in S$ we have $\mu(A)+\mu(B)=\mu(A \cup B)$ for all $A,B \in Q$. From this, lemma \ref{l:nohigher} gives $\sum_{A \in Q} \mu(A) = \mu(\Xi) =1 $. Subtracting $\mu(R)$, and remembering that the quantum measure is non-negative, we have that $\sum_{A \in S} \mu(A) \leq 1$. Using (\ref{e:q_dec}), this establishes that CE holds for the probability function $P$, proving the theorem.
\end{proof}
This allows a number of interesting results to be imported into quantum measure theory from the study of local orthogonality and CE, of which the following are instructive and representative but certainly not exhaustive (see \cite{Fritz:2013,Yan:2013,Amaral:2014,Sainz:2014} for more).
\begin{corollary}
The following are properties of all probability functions on partition scenarios that admit a joint quantum measure:
\begin{itemize}
\item[(i)] They imply the quantum bound, $\sqrt{5}$, on the maximum violation of the KCBS inequality for two independent copies of the Wright pentagon scenario;
\item[(ii)] for the CHSH scenario, the existence of two independent copies of this probability function with maximum violation of the CHSH inequality of more than 2.883 is banned; 
\item[(iii)] they allow no advantage over classical (non-contextual) probability functions for the Guess Your Neighbour's Input Game.
\end{itemize}
\end{corollary}
\begin{proof}
As noted above, the Wright pentagon can be constructed in a partition scenario. Thus \textit{(i)} can be proved by combining theorem \ref{t:main} with the arguments in \cite{Cabello:2012}. Bell scenarios and copies thereof are also partition scenarios \cite{Dowker:2013}, and so \textit{(ii)} and \textit{(iii)} can be proved by combining theorem \ref{t:main} with the argument in section 4.3 of \cite{Sainz:2014}, and the first proof in the Methods section of \cite{Fritz:2012}, respectively.
\end{proof}
By construing the principle more broadly (by assuming that certain contextuality scenarios are realisable, or that quantum probability functions must be in the physical set) CE can be made to imply both Tsirelson's bound for CHSH \cite{Cabello:2014a} and the quantum bound for all contextuality scenarios \cite{Amaral:2014} \footnote{This statement of the result is not universally accepted due to differences the framework employed by \cite{Amaral:2014} as opposed to that of \cite{Acin:2012}.  What is called the quantum set in the former framework corresponds, in the latter, only to their $\mathcal{Q}^1$ set, which is larger than their quantum set.}.

As noted above, other definitions of ``lack of third order interference'' have been made.  Finding out whether these versions of the principle are equivalent to the one given here is important for the goal of simplifying and clarifying the list of candidate principles.  Also, if the definition given above implies any of the others, then the results given above can be extended to these other formalisms.  This not ruled out for \cite{Ududec:2011}: while CE is shown to follow from two other assumptions unrelated to third order interference in this formalism, this does not mean that lack of third order interference alone fails to imply CE.  It will require more work to see if the definition of lack of third order interference given here is equivalent to that of \cite{Ududec:2011}, as the formalisms are quite different, and it is non-trivial to embed one formalism in the other.  It is possible that the definitions are only equivalent under some assumptions.  Similarly, it is not obvious that the definition of third order interference given in \cite{Ududec:2011} implies consistent exclusivity by a similar argument to that given above; indeed this implication may be false in general.  \footnote{Note that these considerations are not of interest for \cite{Niestegge:2011, Niestegge:2013} because, as already noted above, in that formalism CE is true for the most general class of models that is defined.}

Many other interesting issues remain open.  Firstly, it would be of great significance if the converse of the above theorem is also true.  However, the construction of a quantum measure from a probability function obeying CE, even if possible, is not a straightforward task.  Secondly, the stronger forms of joint quantum measure considered in \cite{Dowker:2013} have been justified by appealing to composability, and so it would be instructive to know if they can be derived from the above principle by adding some simple  assumptions.  Similarly, it has been asked whether local orthogonality can be strengthened by the addition of further strongly-motivated conditions.  In the light of the results above, work on either one of these questions can now inform the other. Hopefully, answering some of these questions will help to clarify what needs to be added to these principles in order to totally characterise quantum non-locality and contextuality.

\paragraph{Acknowledgements}  This work was made possible through the support of a grant from the John Templeton Foundation and was also supported by the EPSRC DIQIP grant and ERC AdG NLST. 

\bibliographystyle{../jhep}
\bibliography{refs0.2}

\providecommand{\href}[2]{#2}\begingroup\raggedright\begin{thebibliography}{10}

\bibitem{Sorkin:1997gi}
R.~D. Sorkin, {\it Forks in the road, on the way to quantum gravity},  {\em
  Int. J. Theor. Phys.} {\bf 36} (1997) 2759--2781,
  [\href{http://xxx.lanl.gov/abs/gr-qc/9706002}{{\tt gr-qc/9706002}}].

\bibitem{Hartle:1998yg}
J.~B. Hartle, {\it {Generalized quantum theory and black hole evaporation}},
  \href{http://xxx.lanl.gov/abs/gr-qc/9808070}{{\tt gr-qc/9808070}}.

\bibitem{Smolin:2006bw}
L.~Smolin, {\it {Could quantum mechanics be an approximation to another
  theory?}},  \href{http://xxx.lanl.gov/abs/quant-ph/0609109}{{\tt
  quant-ph/0609109}}.

\bibitem{Tsirelson:1980}
B.~Cirel'son, {\it {Quantum generalisations of Bell's inequality}},  {\em Lett.
  Math. Phys} {\bf 4} (1980) 93--100.

\bibitem{Tsirelson:1987}
B.~Tsirel'son, {\it {Quantum analogues of the Bell inequalities: the case of
  two spatially separated domains}},  {\em J. Soviet Math.} {\bf 36} (1987)
  557--570. Translated from a source in Russian of 1985.

\bibitem{Tsirelson:1993}
B.~Tsirelson, {\it {Some results and problems on quantum Bell-type
  inequalities}},  {\em Hadronic Journal Supplement} {\bf 8} (1993) 329--345.

\bibitem{Popescu:1994}
S.~Popescu and D.~Rohrlich, {\it {Nonlocality as an axiom}},  {\em Foundations
  of Physics} {\bf 24} (1994) 379.

\bibitem{vanDam:2005}
W.~van Dam, {\em Nonlocality and Communication Complexity}.
\newblock PhD thesis, University of Oxford, 2000.
\newblock \href{http://xxx.lanl.gov/abs/quant-ph/0501159}{{\tt
  quant-ph/0501159}}.

\bibitem{Craig:2006ny}
D.~A. Craig, H.~F. Dowker, J.~Henson, M.~S., D.~Rideout, and R.~D. Sorkin, {\it
  {A Bell Inequality Analog in Quantum Measure Theory}},  {\em J. Phys.} {\bf
  A40} (2007) 501--523, [\href{http://xxx.lanl.gov/abs/quant-ph/0605008}{{\tt
  quant-ph/0605008}}].

\bibitem{Linden:2007}
N.~Linden, S.~Popescu, A.~J. Short, and A.~Winter, {\it Quantum nonlocality and
  beyond: Limits from nonlocal computation},  {\em Phys. Rev. Lett.} {\bf 99}
  (Oct, 2007) 180502.

\bibitem{Pawlowski:2009}
M.~{Paw{\l}owski}, T.~{Paterek}, D.~{Kaszlikowski}, V.~{Scarani}, A.~{Winter},
  and M.~{{\.Z}ukowski}, {\it {Information causality as a physical principle}},
   {\em Nature} {\bf 461} (Oct., 2009) 1101--1104,
  [\href{http://xxx.lanl.gov/abs/0905.2292}{{\tt arXiv:0905.2292}}].

\bibitem{navascues:2009}
M.~{Navascu\'es} and H.~{Wunderlich}, {\it {A glance beyond the quantum
  model}},  {\em Royal Society of London Proceedings Series A} {\bf 466} (Nov.,
  2009) 881--890, [\href{http://xxx.lanl.gov/abs/0907.0372}{{\tt
  arXiv:0907.0372}}].

\bibitem{Fritz:2012}
T.~Fritz, A.~B. Sainz, R.~Augusiak, J.~B. Brask, R.~Chaves, A.~Leverrier, and
  A.~Ac{\'i}n, {\it Local orthogonality as a multipartite principle for quantum
  correlations},  {\em Nature Communications} {\bf 4} (2013) 2263,
  [\href{http://xxx.lanl.gov/abs/1210.3018}{{\tt arXiv:1210.3018}}].

\bibitem{Cabello:2012}
A.~{Cabello}, {\it Simple explanation of the quantum violation of a fundamental
  inequality},  {\em Physical Review Letters} {\bf 110} (Feb., 2013) 060402,
  [\href{http://xxx.lanl.gov/abs/1210.2988}{{\tt arXiv:1210.2988}}].

\bibitem{Navascues:2014}
M.~Navascu\'es, Y.~Guryanova, M.~J. Hoban, and A.~Ac\'in, {\it {Almost quantum
  correlations}},  {\em Nature Communications} {\bf 6} (2015) 6288,
  [\href{http://xxx.lanl.gov/abs/1403.4621}{{\tt arXiv:1403.4621}}].

\bibitem{Foulis:1981}
D.~J. Foulis and C.~H. Randall, {\it Empirical logic and tensor products},  in
  {\em Interpretations and Foundations of Quantum Theory} (H.~Neumann, ed.),
  pp.~216--233.
\newblock Wissenschaftsverlag, Bibliographisches Institut, Band 5,
  Mannheim/Wien/Zurich, 1981.

\bibitem{Fritz:2013}
T.~{Fritz}, A.~{Leverrier}, and A.~{Bel{\'e}n Sainz}, {\it {A combinatorial
  approach to nonlocality and contextuality}},  {\em ArXiv e-prints} (Dec.,
  2012) [\href{http://xxx.lanl.gov/abs/1212.4084}{{\tt arXiv:1212.4084}}].

\bibitem{Sorkin:1994}
R.~D. Sorkin, {\it Quantum mechanics as quantum measure theory},  {\em Mod.
  Phys. Lett.} {\bf A9} (1994) 3119--3128,
  [\href{http://xxx.lanl.gov/abs/gr-qc/9401003}{{\tt gr-qc/9401003}}].

\bibitem{Sinha:2009}
U.~Sinha, C.~Couteau, Z.~Medendorp, I.~Sollner, R.~Laflamme, R.~Sorkin, and
  G.~Weihs, {\it {Testing Born’s rule in Quantum mechanics with a triple slit
  experiment}},  in {\em {Foundations of Probability and Physics
  \textrm{\textbf{5}}, American Institute of Physics Conference Proceedings
  \textbf{Vol. 1101}}}, pp.~200--207.
\newblock New York, 2009.
\newblock \href{http://xxx.lanl.gov/abs/0811.2068}{{\tt arXiv:0811.2068}}.

\bibitem{Ududec:2011}
C.~Ududec, H.~Barnum, and J.~Emerson, {\it Three slit experiments and the
  structure of quantum theory},  {\em Foundations of Physics} {\bf 41} (2011),
  no.~3 396--405, [\href{http://xxx.lanl.gov/abs/0909.4787}{{\tt 0909.4787}}].

\bibitem{Barnum:2014ysa}
H.~Barnum, M.~P. Mueller, and C.~Ududec, {\it {Higher-order interference and
  single-system postulates characterizing quantum theory}},
  \href{http://xxx.lanl.gov/abs/1403.4147}{{\tt arXiv:1403.4147}}.

\bibitem{Niestegge:2011}
G.~{Niestegge}, {\it {Three-Slit Experiments and Quantum Nonlocality}},  {\em
  Foundations of Physics} {\bf 43} (June, 2013) 805--812,
  [\href{http://xxx.lanl.gov/abs/1104.0091}{{\tt arXiv:1104.0091}}].

\bibitem{Niestegge:2013}
G.~Niestegge, {\it Super quantum probabilities and three-slit
  experiments—{W}right's pentagon state and the {P}opescu–{R}ohrlich box
  require third-order interference},  {\em Phys. Scr.} (2014) 014034,
  [\href{http://xxx.lanl.gov/abs/1303.3986}{{\tt arXiv:1303.3986}}].

\bibitem{Note1}
\label {f:other_formalisms}The results reported here differ substantially from
  Niestegge's, and the use of third order interference is quite different.
  Here, as far as the ban on PR boxes and Wright pentagon states is concerned,
  the only role of lack of third order interference is to imply CE. Then the
  result follows from CE and a composibility assumption as in \cite
  {Fritz:2012, Cabello:2012}. Niestegge implicitly assumes CE as part of his
  basic framework (see \protect \textit {e.g.} page 3 of \cite
  {Niestegge:2012}, where the comments on the properties of states on
  ``orthospaces'' easily imply CE), and the same limitations on behaviours are
  derived without the composibility assumption, by adding a version of lack of
  third order interference. The fact that these important properties of QM can
  be derived in two quite different starting points, one involving independent,
  separated systems and the other not, is surely significant for
  reconstructions of QM from simple principles.

\bibitem{Dowker:2013}
F.~Dowker, J.~Henson, and P.~Wallden, {\it {A histories perspective on
  characterizing quantum non-locality}},  {\em New J.Phys.} {\bf 16} (2014)
  033033, [\href{http://xxx.lanl.gov/abs/1311.6287}{{\tt arXiv:1311.6287}}].

\bibitem{Abransky:2011}
S.~Abramsky and A.~Brandenburger, {\it The sheaf-theoretic structure of
  non-locality and contextuality},  {\em New Journal of Physics} {\bf 13}
  (2011), no.~11 113036, [\href{http://xxx.lanl.gov/abs/1102.0264}{{\tt
  arXiv:1102.0264}}].

\bibitem{Note2}
Indeed these ``correlated'' or ``branching'' measurements (see \cite
  {Dowker:2013} and appendix D of \cite {Fritz:2013}) must be added if the
  definition of exclusivity given above is to make intuitive sense; otherwise
  it is possible that two outcomes which contradict each other for a particular
  box may not appear as exclusive alternatives in one measurement.
  Alternatively, exclusivity can be defined so that outcomes $A$ and $B$ are
  exclusive if there exists a course-grained outcomes $C$ and $D$ such that $A
  \subseteq C$, $B \subseteq D$, $C$ and $D$ are disjoint and there exists a
  measurement $M \in \protect \mathcal {M}$ such that $C,D \in M$.

\bibitem{Specker:1960}
E.~Specker, {\it Die logik nicht gleichzeitig entscheidbarer aussagen},  {\em
  Dialectica} {\bf 14} (1960) 239--246. {English translation: “The logic of
  propositions which are not simultaneously decidable”, Reprinted in C. A.
  Hooker (ed.), \textit{The Logico-Algebraic Approach to Quantum Mechanics.
  Volume I: Historical Evolution} (Reidel, Dordrecht, 1975), pp. 135-140.}

\bibitem{Wright:1978}
R.~Wright in {\em {Mathematical Foundations of Quantum Mechanics}}, p.~255.
\newblock San Diego, 1978.

\bibitem{Note3}
Other frameworks inspired by the Kochen-Specker theorem can be defined, for
  example the AFLS formalism \cite {Fritz:2013}. These cases can be treated in
  the marginal scenario framework, and thus represented as partition scenarios,
  via a restriction on probability functions (see section 7 of \cite
  {Abransky:2011} and appendix D of \cite {Fritz:2013}). Using this map between
  the frameworks, it is not difficult to see that CE for the underlying
  marginal scenario implies CE for the derived ALFS scenario, and thus the main
  result of this letter extends to the AFLS framework.

\bibitem{Cabello:2014a}
A.~Cabello, {\it Exclusivity principle and the quantum bound of the bell
  inequality},  {\em Phys. Rev. A} {\bf 90} (Dec, 2014) 062125,
  [\href{http://xxx.lanl.gov/abs/1406.5656}{{\tt arXiv:1406.5656}}].

\bibitem{Yan:2013}
B.~Yan, {\it Quantum correlations are tightly bound by the exclusivity
  principle},  {\em Phys. Rev. Lett.} {\bf 110} (Jun, 2013) 260406,
  [\href{http://xxx.lanl.gov/abs/1303.4357}{{\tt arXiv:1303.4357}}].

\bibitem{Amaral:2014}
B.~Amaral, M.~T. Cunha, and A.~Cabello, {\it Exclusivity principle forbids sets
  of correlations larger than the quantum set},  {\em Phys. Rev. A} {\bf 89}
  (Mar, 2014) 030101, [\href{http://xxx.lanl.gov/abs/1306.6289}{{\tt
  arXiv:1306.6289}}].

\bibitem{Sainz:2014}
A.~B. Sainz, T.~Fritz, R.~Augusiak, J.~B. Brask, R.~Chaves, A.~Leverrier, and
  A.~Ac\'in, {\it Exploring the local orthogonality principle},  {\em Phys.
  Rev. A} {\bf 89} (Mar, 2014) 032117,
  [\href{http://xxx.lanl.gov/abs/1311.6699}{{\tt arXiv:1311.6699}}].

\bibitem{Note4}
This statement of the result is not universally accepted due to differences the
  framework employed by \cite {Amaral:2014} as opposed to that of \cite
  {Acin:2012}. What is called the quantum set in the former framework
  corresponds, in the latter, only to their $\protect \mathcal {Q}^1$ set,
  which is larger than their quantum set.

\bibitem{Note5}
Note that these considerations are not of interest for \cite {Niestegge:2011,
  Niestegge:2013} because, as already noted above, in that formalism CE is true
  for the most general class of models that is defined.

\bibitem{Niestegge:2012}
G.~Niestegge, {\it Conditional probability, three-slit experiments, and the
  {J}ordan algebra structure of quantum mechanics},  {\em Advances in
  Mathematical Physics} {\bf 2012} (2012)
  [\href{http://xxx.lanl.gov/abs/0912.0203}{{\tt arXiv:0912.0203}}]. Article ID
  156573.

\bibitem{Acin:2012}
A.~Ac\'{\i}n, T.~Fritz, A.~Leverrier, and A.~B. Sainz, {\it {A Combinatorial
  Approach to Nonlocality and Contextuality}},  2012.
\newblock arXiv preprint arXiv:1212.4084.

\end{thebibliography}\endgroup
\end{document}